\documentclass[onecolumn]{IEEEtran}
\usepackage{amsmath,amssymb,epsf,epsfig,amsthm,times,ifthen,epic,psfrag}
\usepackage{enumerate}
\usepackage{cite}    
\usepackage{pdfsync}

\begin{document}

\newcommand{\creationtime}{\today \ \ @ \theampmtime}


\renewcommand{\qedsymbol}{$\blacksquare$} 


\newtheorem{theorem}              {Theorem}     [section]
\newtheorem{lemma}      [theorem] {Lemma}
\newtheorem{corollary}  [theorem] {Corollary}
\newtheorem{proposition}[theorem] {Proposition}

\theoremstyle{remark}
\newtheorem{remark}     [theorem] {Remark}
\newtheorem{algorithm}  [theorem] {Algorithm}
\newtheorem{conjecture} [theorem] {Conjecture}

\theoremstyle{definition}         
\newtheorem{definition} [theorem] {Definition}
\newtheorem{example}    [theorem] {Example}
\newtheorem*{claim}  {Claim}
\newtheorem*{notation}  {Notation}
\newcommand{\Comment}[1]{& [\mbox{from  #1}]}
\newcommand{\Commenta}[2]{&[\mbox{since #2}]}
\newcommand{\mc}{\mathcal}
\newcommand{\mb}{\mathbf}
\newcommand{\abs}[1]{\left\lvert #1 \right\rvert}
\newcommand{\card}[1]{\abs{#1}}
\newcommand{\Network}{\mathcal{N}}
\newcommand{\node}{v}
\newcommand{\edge}{e}
\newcommand{\nodes}{\mathcal{V}}
\newcommand{\vertices}[1]{\nodes}
\newcommand{\inNodes}[1]{V_{#1}}
\newcommand{\Capacity}[1]{w(#1)}
\newcommand{\edges}{\mathcal{E}}
\newcommand{\inEdges}[1]{\mathcal{E}_{in}(#1)}
\newcommand{\outEdges}[1]{\mathcal{E}_{out}(#1)}
\newcommand{\encodingFunction}[2]{h_{#1,#2}}
\newcommand{\Dimension}[1]{\text{Dim}(#1)}
\newcommand{\TransVar}[1]{\hat{#1}}
\newcommand{\RecVar}[1]{\tilde{#1}}
\newcommand{\Code}{\mathcal{C}}
\newcommand{\Directed}[1]{\vec{#1}}
\newcommand{\alphabet}{\mathcal{A}}
\newcommand{\sources}{S}
\newcommand{\decodeAlphabet}{\mathcal{B}}
\newcommand{\delay}{\delta}
\newcommand{\integer}{\mathbb{Z}}
\newcommand{\remove}[1]{}
\newcommand{\removeTrue}[1]{}
\newcommand{\sourceSymbol}{\alpha}
\newcommand{\sourceVec}[1]{\sourceSymbol\!\left(#1\right)}
\newcommand{\sumVec}{p}
\newcommand{\edgeVar}[1]{z_{#1}}
\newcommand{\edgeSet}{C}
\newcommand{\cut}[1]{{min-cut{$\left(#1\right)$}}}
\newcommand{\cuts}[1]{\Lambda({#1})}
\newcommand{\minCut}[1]{{min-cut{$\left(#1\right)$}}}
\newcommand{\maxRange}[1]{\hat{R}_{#1}}
\newcommand{\maxRangeB}[1]{R_{#1}}
\newcommand{\receiver}{\rho}
\newcommand{\source}{\sigma}
\newcommand{\decodFunct}{\psi}
\newcommand{\setOfMessages}{x}
\newcommand{\vecInst}{\mathbf{\alpha}}
\newcommand{\messageVecInst}{w}
\newcommand{\encodeMatrix}[1]{M_{#1}}
\newcommand{\VecComp}[2]{#1_{#2}}
\newcommand{\Tree}{\mathcal{T}}
\newcommand{\tree}{T}
\newcommand{\treeIndex}[1]{J_{#1}}
\newcommand{\gap}[1]{\mathcal{C}_{\mbox{{\scriptsize gap}}}\!\left(#1\right)}
\newcommand{\codCap}[1]{\mathcal{C}\!\left(#1\right)}
\newcommand{\linCodCap}[1]{\mathcal{C}_{\mbox{{\scriptsize lin}}}\!\left(#1\right)}
\newcommand{\scalLinCodCap}[1]{\mathcal{C}_{\mbox{{\scriptsize s-lin}}}\!\left(#1\right)}
\newcommand{\routCap}[1]{\mathcal{C}_{\mbox{{\scriptsize rout}}}\!\left(#1\right)}
\newcommand{\edgeFunct}[1]{h^{(#1)}}

\newcommand{\NbinaryBlocks}[1]{{#1}^{(M)}}
\newcommand{\sumset}[1]{\mbox{{\it sum}}\!\left(#1\right)}
\newcommand{\zeroAt}[1]{h^{(#1)}}
\newcommand{\invZeroAt}[1]{h_{#1}^{-1}}
\newcommand{\maxWeight}[1]{\mbox{HW}_{\mbox{\tiny max}}\left(#1\right)}
\newcommand{\cardConst}[1]{\gamma(#1)}
\newcommand{\hammingWeight}[1]{\mbox{HW$\left(#1\right)$}}
\newcommand{\cardsources}{s}
\newcommand{\NOPROCESS}[1]{}
\newcommand{\FunIdentity}{f_{id}}
\newcommand{\FunMajority}{f_{maj}}
\newcommand{\FunSum}{f_{sum}}
\newcommand{\FunParity}{f_{parity}}
\newcommand{\FunMaximum}{f_{max}}
\newcommand{\FunMinimum}{f_{min}}
\newcommand{\footprintsize}{footprint size }
\newcommand{\footprintsizes}{footprint sizes }
\newcommand{\IndicesToIndex}{h}
\newcommand{\steinerNumber}[1]{\Pi\!\left(#1\right)}
\newcommand{\lb}[1]{l\!\left(#1\right)}
\newcommand{\sourceVecList}[2]{\sourceSymbol\!\left(#1\right)_{#2}}
\newcommand{\encodSymbol}{x}
\newcommand{\encodCoeff}[1]{\encodSymbol_{#1}}
\newcommand{\edgeCodSymbol}{x}
\newcommand{\edgeCoeff}[1]{\edgeCodSymbol_{#1}}
\newcommand{\decodSymbol}{x}
\newcommand{\decodCoeff}[1]{\decodSymbol_{#1}}

\newcommand{\encodCoeffVar}[1]{x_{#1}}
\newcommand{\edgeCoeffVar}[1]{x_{#1}}
\newcommand{\decodCoeffVar}[1]{x_{#1}}

\newcommand{\head}[1]{\textit{head($#1$)}}
\newcommand{\tail}[1]{\textit{tail($#1$)}}
\newcommand{\field}[1]{\mathbb{F}_{#1}}
\newcommand{\ring}[2]{\mathbb{F}_{#1}[#2]}
\newcommand{\receivedVec}{\decodFunct}
\newcommand{\ceil}[1]{\left\lceil #1 \right\rceil}
\newcommand{\mVec}{w}
\newcommand{\code}{\mathcal{C}}
\newcommand{\intAdd}{+}
\newcommand{\modqAdd}{\oplus}
\newcommand{\fieldAdd}{\boxplus}
\newcommand{\linearCode}[1]{$#1$-linear}
\newcommand{\ringOnA}{}
\newcommand{\polyRing}[2]{#1\!\left[ #2 \right]}
\newcommand{\zeroVecOver}[1]{\mb{0}_{#1}}
\newcommand{\zeroOver}[1]{0_{#1}}
\newcommand{\Fclass}{\mathcal{D}}
\newcommand{\define}[1]{{\it #1}}
\newcommand{\sourceSet}[1]{\indexSet_{#1}}
\newcommand{\rank}[1]{\mbox{{ rank}}\!\left(#1\right)}
\newcommand{\transferM}{T}
\newcommand{\ideal}[1]{\left\langle #1\right\rangle}
\newcommand{\grobner}{Grob$\ddot{\mbox{o}}$ner }
\newcommand{\gbasis}[1]{\mathcal{G}\!\left(#1\right)}
\newcommand{\LINENET}{1}
\newcommand{\RBF}{2}
\newcommand{\GBASISFIG}{1}
\newcommand{\GBASISFIGEQ}{4}
\newcommand{\variety}{\mathcal{V}}
\newcommand{\idealSet}{\mathcal{I}}
\newcommand{\rad}[1]{#1^{\mbox{{\footnotesize rad}}}}
\newcommand{\alg}{\mbox{{\footnotesize alg}}}
\newcommand{\indexSet}{K}
\newcommand{\monomialSet}{\mathcal{M}}
\newcommand{\lex}{>_{\mbox{\footnotesize{{\it lex}}}}}
\newcommand{\glex}{>_{\mbox{\footnotesize{{\it grlex}}}}}
\newcommand{\mdeg}[1]{\mbox{{\it \footnotesize{multideg$(#1)$}}}}
\newcommand{\pring}[2]{\field{#1}[x_1,x_2,\ldots,x_{#2}]}



\title{Computing linear functions by linear coding over networks 
\thanks{This work was supported by the National Science Foundation award CNS 0916778
        and the UCSD Center for Wireless Communications.\newline
\indent The authors are with the Department of Electrical and Computer Engineering, 
        University of California, San Diego, La Jolla, CA 92093-0407. \ \ 
 (rathnam@ucsd.edu, massimo@ece.ucsd.edu)
}}

\author{Rathinakumar Appuswamy, Massimo Franceschetti}


\maketitle

\begin{abstract}
We consider the scenario in which a set of sources generate messages in a network and a receiver node
demands an arbitrary \emph{linear function} of these messages. 
We formulate an algebraic test to determine
whether an arbitrary network can compute linear functions using \emph{linear codes}.
We identify a class of linear functions that can be computed using linear codes in every network that satisfies a natural
cut-based condition.
Conversely, for another class of linear functions, we show that the cut-based condition 
does not guarantee the existence of a linear coding solution.
For linear functions over the binary field, the two classes are complements of each other. 

\end{abstract}

\thispagestyle{empty}


\section{Introduction}
In many practical networks, including sensor networks and vehicular networks, 
receivers demand a function of the messages generated by the sources that are distributed
across the network rather than the generated messages. This  situation is studied in the framework of network computing \cite{computing1,computing2,ramam,RaiDey2009,Ma,Nazer,Kumar1}.
The classical network coding model of Ahlswede, Cai, Li, and Yeung~\cite{Ahlswede-Cai-Li-Yeung-IT-Jul00} can be viewed as a the special case of network computing in which the  function to be computed at the receivers corresponds to a subset of the source messages and communication occurs over a  network with noiseless links. 

In the same noiseless set up of~\cite{Ahlswede-Cai-Li-Yeung-IT-Jul00}, we consider the scenario in which
a set of source nodes 
generate messages over a finite field and a single receiver node computes a
linear function  of these messages.
We ask whether this linear function can be computed by performing linear coding operations
at the intermediate nodes. 

In multiple-receiver networks, if each receiver node demands a subset of
the source messages (which is an example of a  linear function),
then Dougherty, Freiling, and  Zeger \cite{ken1}  showed  that linear codes are not sufficient to recover the source messages.
Similarly, if each  receiver node demands the sum of the source messages,
then Ray and Dei~\cite{RaiDey2009} showed that linear  codes are also not sufficient to recover the source messages.
In contrast, in single-receiver networks linear  codes are sufficient 
for both the above problems and
a simple cut-based condition can be used to test whether a linear solution exists. 

Our contribution is as follows.
We extend above results investigating if a similar cut-based condition guarantees the existence of a linear solution
when the receiver node demands an arbitrary linear function of the source messages. We identify two classes of functions, one for which the cut-based condition is sufficient for solvability and the other for which it is not. These classes are complements of each other when the source messages are over the binary field.
Along the way, we develop an algebraic framework to study linear codes and provide an algebraic condition to test whether a linear solution exists, similar to the one given by Koetter and M\'{e}dard~\cite{Koetter-Medard-IT-Oct03} for classical network coding.

The paper is organized as follows. We formally introduce the network computation model in Section~\ref{Sec:intro}. 
In Section~\ref{Sec:main} we develop the necessary algebraic tools to study linear codes and introduce the cut-based condition. 
In Section~\ref{Sec:linearFunctions}, we show the main results for the two classes of functions.
Section~\ref{Sec:conclusions} concludes the paper, mentioning some open problems.
\subsection{Network model and preliminaries} \label{Sec:intro}
In this paper, a \textit{network} $\Network$ consists of a finite,
directed acyclic multigraph $G= (\nodes,\edges)$, 
a set of \textit{source nodes} $\sources = \{\source_1, \dots, \source_\cardsources \} \subseteq \ \nodes$, 
and a \textit{receiver} $\receiver \in \nodes$. 
Such a network is denoted by $ \Network = (G,\sources,\receiver)$.
We use the word ``graph" to mean a multigraph, 
and ``network" to mean a single-receiver network.
We assume that  $\receiver \notin \sources$, 
and that the graph
$G$ contains a directed path from every node in $\nodes$ to the receiver $\receiver$. 
For each node $u \in \nodes$, 
let $\inEdges{u}$ and $\outEdges{u}$ denote the in-edges and out-edges of $u$ respectively. 
We  also assume
(without loss of generality)
that if a network node has no in-edges,
then it is a source node.
We use $s$ to denote the number of sources $\card{S}$ in the network.

An \textit{alphabet} $\alphabet$ is a nonzero finite field.
For any positive integer $m$,
any vector $x \in \alphabet^{m}$, 
and any $i$,
let $\VecComp{x}{i}$ denote the $i$-th component of $x$.
For any index set $\indexSet = \{i_1,i_2,\ldots,i_q\} \subseteq \{1,2,\ldots,m\}$ with $i_1 < i_2 < \ldots < i_q$, let $\VecComp{x}{\indexSet}$ denote the vector 
$(\VecComp{x}{i_1},\VecComp{x}{i_2},\ldots,\VecComp{x}{i_q}) \in \alphabet^{\card{\indexSet}}$. 

The 
\textit{network computing} problem consists of a network $\Network$, a source alphabet $\alphabet$, and a 
\define{target function}
$$
f \ : \ \alphabet^s \longrightarrow \decodeAlphabet
$$
where $\decodeAlphabet$ is the \define{decoding alphabet}.
A target function $f$ is \define{linear} if there exists a matrix $T$ over $\alphabet$ such that
$$
f(x) = T x^{t}, \quad \forall \ x \in \alphabet^s
$$
where `$t$' denotes matrix transposition. For linear target functions the decoding alphabet is of the form $\alphabet^l$, with $1\leq l \leq s$.
Without loss of generality, we  assume
that $T$ is full rank (over $\alphabet$) and has no zero columns.
For example, if $T$ is the $s \times s$ identity matrix, then the receiver demands the complete
set of source messages, and this corresponds to the  classical network coding problem.
On the other hand, if $T$ is the row vector of $1$'s, then the receiver demands a sum (over $\alphabet$)
of the source values. 
Let $n$ be a positive integer. 
Given a network $\Network$ with source set $\sources$ and 
alphabet $\alphabet$, a \textit{message generator}
is a mapping
$$\sourceSymbol \ : \ \sources \longrightarrow \alphabet^n.$$
For each source $\source_i \in \sources$, 
$\sourceVec{\source_i}$ is called a \textit{message vector} and it can be viewed as an 
element of $\field{q^n}$ (rather than as a vector).
%
\begin{definition}
A \textit{linear network code} in a network $\Network$ consists of the following: 
\begin{itemize}
\item[(i)]
Every edge $\edge \in \edges$ \textit{carries an element}  of $\field{q^n}$ and this element is denoted by
$z_{\edge}$. 
For any node $\node \in \nodes - \receiver$ and any out-edge $\edge \in \outEdges{\node}$, 
the network code specifies an {\it encoding function} 
$h^{(e)}$ of the form: 
\begin{align} \label{Eq:encoding}
h^{(e)}
&= 
\begin{cases}
\encodCoeff{1,e} \sourceVecList{u}{} + 
\displaystyle \sum_{\hat{e} \in \inEdges{u}} \edgeCoeff{\hat{e},e} \edgeVar{\hat{e}} & \; \text{if} \; u \in \sources \\
\displaystyle \sum_{\hat{e} \in \inEdges{u}} \edgeCoeff{\hat{e},e} \edgeVar{\hat{e}} & \; \text{otherwise}
\end{cases}
\end{align}
where $\edgeCoeff{\hat{e},e}, \encodCoeff{1,e} \in \field{q^n}$ for all $\hat{e} \in \inEdges{u}$.
\item[(ii)] 
%
The {\it decoding function} $\decodFunct$ outputs a 
vector of length $l$ whose $j$-th component is of the form:
\begin{align} \label{Eq:decoding}
\displaystyle \sum_{e \in \inEdges{\receiver}} \decodCoeff{e,j} \edgeVar{e}
\end{align}
where $\decodCoeff{e,j} \in \field{q^n}$ for all $e \in \inEdges{\receiver}$.
The arithmetic in \eqref{Eq:encoding} and \eqref{Eq:decoding} is performed over 
$\field{q^n}$.
\end{itemize}
\end{definition}
%
In this paper, by a {\it network code}, we always mean a linear network code.
In the literature, the class of network codes we define here is referred to as  {\it scalar linear codes.} These codes were introduced and studied in 
\cite{Koetter-Medard-IT-Oct03}.  A more general class of linear codes over $\field{q^n}$ were defined and studied in \cite{ken1,ken2}.

Depending on the context, we may
view  $z_{\edge}$ as a  vector of length-$n$
over $\field{q}$ or as an element of $\field{q^n}$.
Without explicit mention, we use the fact 
that the addition of $a,b \in \field{q^n}$ as elements of a finite 
field coincides with their sum as elements of a vector space over 
$\field{q}$.
Furthermore, we also view $\field{q}$ as a subfield of $\field{q^n}$ 
without explicitly stating the inclusion map.
Let $z_{e_1}, z_{e_2}, \ldots, z_{e_{\card{\inEdges{\receiver}}}}$ denote the vectors carried by the in-edges of the receiver. 
\begin {definition} \label{Def:sol}
A linear network code over $\field{q^n}$ is called  \textit{a linear solution for computing $f$ in $\Network$} (or simply a {\it linear solution} if $f$ and $\Network$ are clear from the context)
if the decoding function $\decodFunct$ is such that for every message generator $\alpha$,
%
\begin{align}
\decodFunct\left(\edgeVar{e_1},\cdots,\edgeVar{e_{\card{\inEdges{\receiver}}}}\right)_j
&=  f\!\left(\sourceVec{\source_1}_j,\cdots,\sourceVec{\source_{\cardsources}}_j\right) \quad \mbox{for all $j \in \{1,2,\ldots,n\}$}. \label{Eq:decodingFunction}
\end{align}
%
\end{definition}
\begin{remark}
Each source  generates $n$ symbols over $\field{q}$ (viewing $\field{q^n}$ as a vector space over $\field{q}$) and the decoder computes  the target function $f$ for each set of
source symbols.
\end{remark}
A set of edges $C \subseteq \edges$ is said to \textit{separate} 
sources $\source_{m_1}, \ldots, \source_{m_d}$
from the receiver $\receiver$, 
if for each $i \in \{1, 2,\ldots, d\}$,
every path from
$\source_{m_i}$ to $\receiver$ contains at least one edge in $C$.
A set $C \in \edges$ is said to be a \textit{cut} if it separates at least one
source from the receiver. Let $\cuts{\Network}$ denote the set of all cuts in network $\Network$.

For any matrix $\transferM \in \field{q}^{l \times s}$,
let $T_i$ denote its $i$-th column.
For an index set $\sourceSet{} \in \{1,2,\ldots,s\}$, 
let $\transferM_{\sourceSet{}}$ denote the 
$l \times \card{K}$ submatrix of $\transferM$ 
obtained by choosing the columns of $\transferM$ indexed by $\sourceSet{}$.
If $C$ is a cut in a network $\Network$, we define the set 
$$
\sourceSet{C} = \{ i \in \sources : \mbox{$C$ disconnects $\source_i$ from $\receiver$}\}.
$$
Finally, for any  network $\Network$ and matrix $T$, we define 
\begin{align} \label{Eq:mincut}
\mbox{\cut{\Network,T}} = \underset{ C \in \cuts{\Network} }{\min} \;\; \frac{ \card{C}}{\rank{\transferM_{\sourceSet{C}} }}.
\end{align}
%

%

%
\section{Algebraic framework} \label{Sec:main}
\subsection{An algebraic test for the existence of a linear solution} \label{Sec:algebra}
Linear solvability for the classical network coding problem was shown to 
be equivalent to the existence of a non-empty algebraic variety in \cite{Koetter-Medard-IT-Oct03}. In the following, we present an analogous characterization for computing linear functions, providing an algebraic test to determine whether a linear solution for computing  a linear function exists. 
The reverse problem of constructing a multiple-receiver network coding (respectively, network computing) problem given an arbitrary set of polynomials,  which is solvable if and only if the corresponding set of polynomials is simultaneously solvable is considered in reference \cite{ken2} 
(respectively, \cite{RaiDey2009}).

We begin by giving some definitions and stating a technical lemma, followed by the main theorem below.

%
For any edge $e = (u,v) \in \edges$, 
let $\head{e} = v$ and $\tail{e} = u$.
Associated with a linear code over $\field{q^n}$,
we define the following three types of matrices: 
\begin{itemize}
\item For each source $\source_\tau \in \sources$, 
define the $1 \times \card{\edges}$ matrix
$A_{\tau}$ as follows:
\begin{align} \label{Eq:MatrixA}
\left(A_\tau \right)_{1,j} = 
\begin{cases}
\encodCoeff{1,e_j}  & \; \text{if} \; e_j \in \outEdges{\source_t} \\
0 									& \; \text{otherwise}.
\end{cases}
\end{align}
\item Similarly define the $l \times \card{\edges}$ matrix $B$
as follows:
\begin{align} \label{Eq:MatrixB}
B_{i,j} = 
\begin{cases}
\decodCoeff{e_j,i}  & \; \text{if} \; e_j \in \inEdges{\receiver} \\
0 							    & \; \text{otherwise}.
\end{cases}
\end{align}
\item Define the $\card{\edges} \times \card{\edges}$
matrix $F$ as follows:
\begin{align} \label{Eq:MatrixF}
F_{i,j} = 
\begin{cases}
\edgeCoeff{e_i,e_j}   & \; \text{if} \; \head{e_i}=\tail{e_j} \\
0 								  	& \; \text{otherwise}.
\end{cases}
\end{align}
\end{itemize}
Since the graph $G$ associated with the network is acyclic, 
we can assume that the edges $e_1,e_2,\ldots$ 
are ordered such that the matrix 
$F$ is strictly upper-triangular.
Let $I$ denote the identity matrix of suitable dimension.
Consider a network $\Network$ 
with alphabet $\field{q}$ 
and consider a linear code over $\field{q^n}$ with associated matrices 
$A_1,A_2,\ldots,A_{\cardsources},B$ and $F$. 
For every $\tau \in \left\{1,2,\ldots,s\right\}$, 
define the $1 \times l$ matrix 
\begin{align} \label{Eq:matrixMt}
M_\tau=A_\tau(I - F)^{-1} B^{t}.
\end{align}
Now let $x_{A}$ be a vector containing all the non-zero entries of the matrices $A_{\tau}, \tau=1,2,\cdots,s$,
and let $x_{B}$ (respectively, $x_F$) be a vector containing all the non-zero entries 
of the matrix $B$ (respectively, $F$).

By abusing notation, depending on the context we may view $x_{e_i,e_j}$, $x_{i,e_j}$, 
$x_{e_i,j}$ as elements of $\field{q^n}$ or 
as indeterminates.
Thus, each of the matrices defined above may either be a matrix over $\field{q^n}$ or 
a matrix 
over the polynomial ring 
$R = \polyRing{\field{q^n}}{x_A,x_F,x_B}$.
The context should make it clear which of these two notions is being referred to at any given point.
%
\begin{lemma} \label{Lemma:Koetter}
The following two statements hold:
\begin{enumerate}
\item The matrix $I-F$ has a polynomial inverse with 
coefficients in  $\polyRing{\field{q^n}}{x_F}$, 
the ring of polynomials in the variables
constituting $x_F$.
\item 
The decoding function can be written as
$$
\displaystyle \sum_{\tau=1}^{s} \sourceVec{\source_{\tau}} A_{\tau} (I - F)^{-1} B^{t}
$$
\remove{
\item If the determinant of the $c \times c$ matrix
\begin{align*}
M = \prod_{\tau=1}^{s} M_{\tau}
\end{align*}
is nonzero over the ring $R$, 
then there exists a finite extension field $\field{q^n}$ of $\field{q}$ which satisfies the property that 
the variables $\encodCoeffVar{j,e},\ldots,\edgeCoeffVar{e',e}$,
and $\decodCoeffVar{e,l}$
can be assigned values over $\field{q^n}$ such that
the determinant of the matrix $M$ is nonzero over
the field $\field{q^n}$.
}
\end{enumerate}
\end{lemma}
\begin{proof}
The first assertion is a restatement of \cite[Lemma~2]{Koetter-Medard-IT-Oct03}
and the second assertion follows from \cite[Theorem~3]{Koetter-Medard-IT-Oct03}.
\end{proof}
\begin{definition}
Let $R$ be a polynomial ring.
The ideal generated by a subset $X \subset R$ and denoted by $\ideal{X}$
is the smallest ideal in $R$ containing $X$.
\end{definition}
Let $\Network$ be a network
with alphabet $\field{q}$.
Let $R = \polyRing{\field{q}}{x_A,x_F,x_B}$ and $T \in \field{q}^{l \times s}$.
Consider a linear network code for computing the linear function corresponding to $T$ in $\Network$
and the associated matrices $M_{\tau},\tau=1,2,\ldots,s$ over $R$
and define 
$$Z_{\tau} = (T_{\tau})^{t} - M_{\tau} \; \mbox{for $\tau = 1,2,\ldots,s$}.$$
Let $J$ denote the ideal generated 
by the elements of $Z_{\tau} \in R^{1 \times l}, \tau=1,2,\ldots,s$ 
in the ring $R$.
More formally, let
$$
J = \ideal{ \left\{ \left\{
\left(Z_{\tau}\right)_{1},
\left(Z_{\tau}\right)_{2},
\ldots,
\left(Z_{\tau}\right)_{l}  \right\}: \, \tau = 1,2,\dots,s \right\} }.
$$
The polynomials $\left(Z_i\right)_{j}$ are referred to as the {\it generating polynomials} of the ideal $J$.
%
We denote the \grobner basis of an ideal generated by subset $X \subset R$ of a 
polynomial ring $R$ by $\gbasis{X}$. 
The following theorem is a consequence of Hilbert Nullstellensatz 
(see \cite[Lemma~VIII.7.2]{Hungerford} and the remark 
after \cite[Proposition~VIII.7.4]{Hungerford}). 
\begin{theorem} \label{Th:gbasisCondition}
Consider a network  $\Network$ with alphabet $\field{q}$
and the linear target function $f$ corresponding to a matrix
$\transferM \in \alphabet^{l\times s}$. 
There exists an $n > 0$ and a linear solution over $\field{q^n}$ for computing
$f$ in $\Network$ if and only if $\gbasis{J} \neq \{1\}$.
\end{theorem}
\begin{proof}
From Lemma~\ref{Lemma:Koetter}, the vector computed at the receiver can be written as
\begin{align} \label{Eq:condition1}
\decodFunct\left(\edgeVar{e_1},\cdots,\edgeVar{e_{\card{\inEdges{\receiver}}}}\right) & = \begin{pmatrix}
M_1^{t} & M_2^{t} & \cdots & M_s^{t} \end{pmatrix} \
\begin{pmatrix}
\sourceVec{\source_1} \\
\sourceVec{\source_2} \\
\vdots \\
\sourceVec{\source_{\cardsources}}
\end{pmatrix}.
\end{align}
On the other hand, to compute the linear function corresponding to $T$, the decoding function must satisfy
\begin{align} \label{Eq:condition2}
\decodFunct\left(\edgeVar{e_1},\cdots,\edgeVar{e_{\card{\inEdges{\receiver}}}}\right)
&= T \ 
\begin{pmatrix}
\sourceVec{\source_1} \\
\sourceVec{\source_2} \\
\vdots \\
\sourceVec{\source_{\cardsources}}
\end{pmatrix}. & \Comment{\eqref{Eq:decodingFunction}}
\end{align}
It follows that the encoding coefficients in a linear solution must be such that
\begin{align}
(T_{\tau})^{t} - M_{\tau} & = 0  \; \mbox{for $\tau = 1,2,\ldots,s$}. & \Comment{\eqref{Eq:condition1} and \eqref{Eq:condition2}}
\end{align}
If we view the coding coefficients as variables, then it follows that a solution must simultaneously
solve the generating polynomials of the corresponding ideal $J$.
By \cite[Lemma~VIII.7.2]{Hungerford}, such a solution exists over the algebraic closure $\bar{\field{q}}$ of $\field{q}$ if and only if
$J \neq \polyRing{\field{q}}{x_A,x_F,x_B}$.
Furthermore, $J \neq \polyRing{\field{q}}{x_A,x_F,x_B}$ if and only if $\gbasis{J} \neq \{1\}$.
Moreover,  a solution exists over the algebraic closure $\bar{\field{q}}$ of $\field{q}$ if and only if
it exists over some extension field $\field{q^n}$ of $\field{q}$ and the proof is now complete.
\end{proof}
\subsection{Minimum cut condition} \label{Sec:cuts}
It is clear that the set of linear functions that can be solved in a network depends on the network topology. It is easily seen that
a linear solution for computing a linear target function corresponding to $T \in \field{q}^{l \times s}$ exists only if 
the network $\Network$ is such that
for every $C \in \cuts{\Network}$, 
the value of the cut $\card{C}$ is at least the rank of the submatrix $T_{K_C}$ (recall that $K_C$ is the index set of the sources separated by the cut $C$).
This observation is stated in the following lemma which is an immediate consequence of the cut-based bound in \cite[Theorem~2.1]{computing1}. 
\begin{lemma} \label{Lemma:ub}
For a network  $\Network$, a necessary condition for the existence of a linear solution  for computing the target
function corresponding to $T \in \field{q}^{l \times s}$ 
is
$$
\mbox{\cut{\Network,T}} \ge 1.
$$
\end{lemma}
We now consider two special cases.
First, consider the case  in which the receiver demands all the source messages. The corresponding $T$ is given 
by the $s \times s$ identity matrix $I$ and  the condition $\mbox{\cut{\Network,T}} \ge 1$
reduces to 
$$
\frac{\card{C}}{\card{K_C}} \ge 1 \quad \forall \ C \in \cuts{\Network}
$$
i.e., the number of edges in the cut be at least equal to the number of sources separated by the cut.
Second, consider the case in which  the receiver demands the sum of the source messages. The corresponding
matrix $T$ is an $1 \times s$ row vector and the requirement that $\mbox{\cut{\Network,T}} \ge 1$ reduces to
$$
\card{C} \ge 1 \quad \forall \ C \in \cuts{\Network}
$$
i.e., all the sources have a directed path to the receiver.
For both of the above  cases, the cut condition in Lemma~\ref{Lemma:ub} is also sufficient for the existence of a solution. This is shown in
\cite[Theorem~3.1 and Theorem~3.2]{computing1} and is reported in the following Lemma:
\begin{lemma} \label{Th:extremeMatrix}
Let $l \in \{1,s\}$.
For a network $\Network$ with the linear target function $f$ corresponding to a matrix $\transferM \in \alphabet^{l\times s}$, 
a linear solution exists if and only if $\mbox{\cut{\Network,T}} \ge 1$.
\end{lemma}
The focus in the rest of the paper is to extend above results to the case
$l \notin \{1,s\}$ by using the algebraic test of Theorem~\ref{Th:gbasisCondition}.
%
\section{Computing linear functions} \label{Sec:linearFunctions}
In the following, we first define an equivalence relation among matrices and then use it to identify a set
of functions 
that are linearly solvable in every network satisfying the 
condition $\mbox{\cut{\Network,T}} \ge 1$. 
We then construct a linear function outside this set, and a corresponding network with $\mbox{\cut{\Network,T}} \ge 1$, on which such a function cannot be computed with linear codes. 
Finally, we use this example as a building block to identify a  set of linear functions for which there exist networks satisfying the min-cut condition and on which these functions are not solvable.

%

Notice that for a linear function with  matrix $T \in \field{q}^{l \times s}$,
 each column of $\transferM$ corresponds to a single source node.
Hence, for every $s \times s$ permutation matrix 
$\Pi$, computing 
$\transferM x$ is equivalent to computing $\transferM  \Pi x$
after appropriately renaming the source nodes.
Furthermore, for every $l \times l$ full rank matrix $Q$ over $\field{q}$,
computing $\transferM x$ is equivalent to computing $Q  \transferM x$.
These observations motivate the following definition:
\begin{definition}
Let $\transferM \in \field{2}^{l \times s}$ and $\transferM' \in \field{2}^{l \times s}$.
We say
$
\transferM \ \sim  \ \transferM'
$
if there exist an invertible matrix $Q$ of size $l \times l$ and a permutation matrix $\Pi$ of size $s \times s$ such that $\transferM = Q  \transferM'  \Pi$,
and 
$
\transferM \ \nsim \ \transferM'
$
if such $Q$ and $\Pi$ do not exist.
\end{definition}
Since $\transferM$ is assumed to be a full rank matrix,
$\Pi$ can be chosen such that 
the first $l$ columns of $\transferM  \Pi$ are linearly independent.
Let $\hat{\transferM}$ denote the first $l$ columns of $\transferM  \Pi$.
By choosing $Q = \hat{\transferM}^{-1}$, we have
$T \sim Q T  \Pi = (I \ P)$ where $P$ is an $l \times s-l$ 
matrix.
So for an arbitrary linear target function $f$ and an associated matrix $T$, 
there exists an $l \times s-l$ 
matrix $P$ such that $T \sim (I \ P)$.
Without loss of generality, 
we  assume that  each column of $\transferM$ associated with a target function is non-zero.

%
\begin{theorem} \label{Th:solvable}
Consider a network $\Network$  with a linear target function corresponding to a matrix $\transferM \in \field{q}^{(s-1) \times s}$ (i.e., $l = s-1$).
If %
$$
\transferM \sim (I \ u)
$$
where $u$ is a column vector of units,
then a  necessary and sufficient condition for the existence of a linear solution is
$\mbox{\cut{\Network,T}} \ge 1$.
\end{theorem}
\begin{proof} 
Let $T = (I \ u)$.
The `necessary' part is clear from Lemma~\ref{Lemma:ub}. We now focus on the `sufficiency' part.
Notice that for each $\tau=1,2,\ldots,s$, the matrix $M_{\tau}$ (computed as in \eqref{Eq:matrixMt}) 
is a row vector of length $s-1$.
Stack these $s$ row vectors to form an $s \times (s-1)$
matrix $M$ as follows,
$$
M = \begin{pmatrix} 
M_1 \\
M_2 \\
\vdots \\
M_s
\end{pmatrix}.
$$
Let $M_{(i)}$ denote the $(s-1)\times (s-1)$ submatrix of $M$
obtained by deleting its $i$-th row. \\
{\it Claim $1$}: The matrix
$$
\prod_{i=1}^{s} M_{(i)}
$$
has a non-zero determinant over the ring
$R = \polyRing{\field{q}}{x_A,x_F,x_B}$. \\
{\it Claim $2$}: For each $i=1,2,\ldots,s-1$, we have $\left(A_s (I-F)^{-1} B^{t} M_{(s)}^{-1}  \right)_i \neq 0$. \\
By Claim 1 and the sparse zeros lemma \cite{Koetter-Medard-IT-Oct03}, \cite{Schwartz}, it follows that 
 that there exists 
some $n > 0$ such that the variables $\edgeCoeffVar{e',e},\decodCoeffVar{e,l}$
can be assigned values over $\field{q^n}$ so that 
the $s \times (s-1)$ matrix 
$$
M = 
\begin{pmatrix} 
A_1 (I-F)^{-1} B^t \\
A_2 (I-F)^{-1} B^t \\
\vdots \\
A_s (I-F)^{-1} B^t
\end{pmatrix}
$$
is such that any of its $(s-1) \times (s-1)$ submatrices
$M_{(i)}, i=1,2,\ldots,s$ obtained by deleting the $i$-th row in $M$,
is full rank over $\field{q^n}$. 
Define two $s-1 \times s-1$ diagonal matrices  $U$ and $D$ such that
for $i \in \{1,2,\cdots,s-1\}$
\begin{align} 
U_{i,i} & = u_i  \nonumber \\
D_{i,i} & = \left(A_s (I-F)^{-1} B^{t} M_{(s)}^{-1} \right)_i. \label{Eq:diagonalMatrices}
\end{align}

Now define the following matrices over $\field{q^n}$:
\begin{align}
\bar{B} & = D^{-1} U (M_{(s)}^{t})^{-1} B \nonumber \\
\bar{A}_i & = u_i^{-1} \left(A_s (I-F)^{-1} \bar{B}^{t}\right)_i A_i \quad i=1,2,\ldots,s-1 \label{Eq:redefineA}\\
\bar{A}_s & = A_s. \nonumber
\end{align}
By by Claim 2 it follows that $D^{-1}$ exists.
If the matrices $\bar{A}_{\tau},F$, and $\bar{B}$ define a linear network code,
then by Lemma~\ref{Lemma:Koetter}, the vector received by $\receiver$  can be written as,
\begin{align} \label{Eq:receivedSum}
\bar{M}^{t} \ \  
\begin{pmatrix} 
\sourceVec{\source_1} \\
\sourceVec{\source_2} \\
\vdots \\
\sourceVec{\source_s}
\end{pmatrix}
\end{align}
where,
%
\begin{align}
\bar{M} & = 
\begin{pmatrix} 
\bar{A}_1 (I-F)^{-1} \bar{B}^{t} \\
\bar{A}_2 (I-F)^{-1} \bar{B}^{t} \\
\vdots \\
\bar{A}_s (I-F)^{-1} \bar{B}^{t}
\end{pmatrix}. \label{Eq:DefMstar}
\end{align}
We have
\begin{align}
\begin{pmatrix} 
A_1 (I-F)^{-1} \bar{B}^{t} \\
A_2 (I-F)^{-1} \bar{B}^{t} \\
\vdots \\
A_s (I-F)^{-1} \bar{B}^{t}
\end{pmatrix} & = 
\begin{pmatrix}
A_1 (I-F)^{-1} (D^{-1} U (M_{(s)}^{t})^{-1} B)^{t}  \\
A_2 (I-F)^{-1} (D^{-1} U (M_{(s)}^{t})^{-1} B)^{t} \\
\vdots \\
A_s (I-F)^{-1} (D^{-1} U (M_{(s)}^{t})^{-1} B)^{t}
\end{pmatrix} & \Comment{$\bar{B} = D^{-1} U (M_{(s)}^{t})^{-1} B$} \nonumber \\
& = 
\begin{pmatrix}
A_1 (I-F)^{-1} B^{t} M_{(s)}^{-1}  \\
A_2 (I-F)^{-1} B^{t} M_{(s)}^{-1} \\
\vdots \\
A_s (I-F)^{-1} B^{t} M_{(s)}^{-1}
\end{pmatrix} \ D^{-1} U   &  \Comment{$\left( (M_{(s)}^{t})^{-1}\right)^{t} = M_{(s)}^{-1}$} \nonumber \\
& = \begin{pmatrix}
		I \\
		A_s (I-F)^{-1} B^{t} M_{(s)}^{-1}
		\end{pmatrix} \ D^{-1} U  & \Comment{construction of $M_{(s)}$} \label{Eq:nulling} \\
\begin{pmatrix} 
\bar{A}_1 (I-F)^{-1} \bar{B}^{t} \\
\bar{A}_2 (I-F)^{-1} \bar{B}^{t} \\
\vdots \\
\bar{A}_s (I-F)^{-1} \bar{B}^{t}
\end{pmatrix} & = \begin{pmatrix}
		U^{-1} D \\
		A_s (I-F)^{-1} B^{t} M_{(s)}^{-1}
		\end{pmatrix} \ D^{-1} U   &  \Comment{\eqref{Eq:redefineA} and \eqref{Eq:nulling} } \nonumber \\
	  & = \begin{pmatrix}
		U^{-1}  \\
		\mb{1}^{t}
		\end{pmatrix} \ U & \Comment{\eqref{Eq:diagonalMatrices}}  \nonumber \\
	  & = \begin{pmatrix}
		I  \\
		\mb{1}^{t} U
		\end{pmatrix}  \nonumber \\
 		& = \begin{pmatrix}
		I \\
		u^{t}
		\end{pmatrix} \label{Eq:SolMstar} \\ 
\bar{M}^{t} & =   \begin{pmatrix} 
 I \ u \end{pmatrix}. & \Comment{\eqref{Eq:DefMstar} and \eqref{Eq:SolMstar}} \label{Eq:solMstar1}
\end{align}
By substituting \eqref{Eq:solMstar1} in \eqref{Eq:receivedSum},
we conclude that the receiver
computes  the desired linear function by employing 
the network code defined by the encoding matrices
$\{\bar{A}_i,i=1,2,\ldots,s\}$, $\bar{B}$, and $F$. 

The proof of the theorem is now complete for the case when $T = (I \ u)$.
If $\transferM \sim (I \ u)$, then there exists a full-rank matrix $Q$
and a column vector $u'$ of non-zero elements over $\field{q}$ such that
\begin{align*}
T & = Q \ (I \ u'). & \Comment{From Lemma~\ref{Lemma:equivalenceLinearFunctions} in the Appendix}
\end{align*}
Since a full-rank linear operator preserves linear-independence
among vectors, 
for every such full-rank matrix $Q$, we have
\begin{align} \label{Eq:rankEq}
 \rank{\transferM_{\sourceSet{C}}} & =  \rank{(Q^{-1}\transferM)_{\sourceSet{C}}} \quad \forall \ C  \in \cuts{\Network}.
\end{align}
Equation \eqref{Eq:rankEq} implies that $ \mbox{\cut{\Network,T}}= \mbox{\cut{\Network,Q^{-1}T}}$.
Since $Q^{-1}T = (I \ u')$,
from the first part of the proof, there exist an $n > 0$ 
and coding matrices $A_{\tau}, \tau=1,2,\cdots,s$, $F$, and $B$ over $\field{q^n}$ 
such that the receiver can compute the linear target function corresponding to 
$(I \ u')$ if and only if $\mbox{\cut{\Network,T}} \ge 1$.
It immediately follows that by utilizing a code corresponding to the coding 
matrices $A_{\tau}, \tau=1,2,\cdots,s$, $F$, and $QB$, the receiver can compute the target 
function corresponding to $Q (I \ u') = T$.

All that remains to be done is to provide proofs of claims $1$ and $2$.\\
{\it Proof of Claim 1}: 
%
%
If a cut $C$ is such that $\card{\sourceSet{C}} \le s-1$,
then 
\begin{align*}
\card{C} & \ge  \rank{\transferM_{\sourceSet{C}} } & \Comment{$\mbox{\cut{\Network,T}} \ge 1$ and \eqref{Eq:mincut}} \\
		 & = \card{\sourceSet{C}}. & \Comment{$T = (I \ u)	$}
\end{align*}
Thus by \cite[Theorem~3.1]{computing1}, 
there exists a routing solution to compute the identity function of the sources $\{\source_i,i \in \sourceSet{C}\}$
at the receiver.
Let $\card{\sourceSet{C}} = s-1$ and let $\sourceSet{C} = \{1,2,\ldots,j-1,j+1,\ldots,s\}$
for some (arbitrary) $j$.
By Lemma~\ref{Lemma:Koetter}, 
after fixing $\sourceVec{\source_j}=0$,
the vector received by $\receiver$ can be written as
$$
M_{(j)}^{t} \ \ 
\begin{pmatrix} 
\sourceVec{\source_1} \\
\sourceVec{\source_2} \\
\vdots \\
\sourceVec{\source_{j-1}} \\
\sourceVec{\source_{j+1}} \\
\vdots \\
\sourceVec{\source_s} 
\end{pmatrix}.
$$
The existence of a routing solution for computing the identity 
function guarantees that there 
exist 
$\edgeCoeffVar{e',e},\decodCoeffVar{e,l} \in \{0,1\}$ such that the 
matrix $M_{(j)}$ has a non-zero determinant over $\field{q}$.
It follows that the
determinant of $M_{(j)}$ is non-zero over $\polyRing{\field{q}}{x_A,x_F,x_B}$.
Since $j \in \{1,2,\ldots,s\}$ was arbitrary in the above argument, 
it follows that the determinant of each  $M_{(j)}, j=1,2,\ldots,s$ is non-zero over 
$\polyRing{\field{q}}{x_A,x_F,x_B}$ and the claim follows. \\
{\it Proof of Claim 2}: We have
\begin{align}
M \ M_{(s)}^{-1} & = \begin{pmatrix} 
A_1 (I-F)^{-1} B^{t} \\
A_2 (I-F)^{-1} B^{t} \\
\vdots \\
A_s (I-F)^{-1} B^{t}
\end{pmatrix} \quad M_{(s)}^{-1} \nonumber \\
& \overset{(a)} {=} 
\begin{pmatrix} 
I \\ 
A_s(I-F)^{-1} B^{t} M_{(s)}^{-1}
\end{pmatrix} 
 \label{Eq:standardize}
\end{align}
where, $(a)$ follows from the definition of $M_{(s)}^{-1}$.
By contraction, assume that there exists an  $i\in \{1,2,\ldots,s-1\}$ such that 
 $\left(A_s(I-F)^{-1}\bar{B}^{t}\right)_i = 0$.
It then follows that 
\begin{align}
 A_s(I-F)^{-1}B^{t} M_{(s)}^{-1} & =  
 \sum_{j=1}^{s-2} \left(A_s(I-F)^{-1}B^{t} M_{(s)}^{-1}\right)_{i_j} (A_{i_j}(I-F)^{-1}B^{t} M_{(s)}^{-1}) & \Comment{\eqref{Eq:standardize}} \label{Eq:dependance}
\end{align}
for some choice of $i_j \in \{1,2,\ldots,s-1\}, j=1,2,\ldots,s-2$ and
\begin{align}
 \Big( A_s(I-F)^{-1}B^{t} - \sum_{j=1}^{s-2} \left(A_s(I-F)^{-1}B^{t} M_{(s)}^{-1}\right)_{i_j} 
 (A_{i_j}(I-F)^{-1}B)^{t} \Big) \; M_{(s)}^{-1} & = 0 & \Comment{\eqref{Eq:dependance}} \nonumber \\
 \Big( A_s(I-F)^{-1}B^{t} - \sum_{j=1}^{s-2} \left(A_s(I-F)^{-1}B^{t} M_{(s)}^{-1}\right)_{i_j} 
 (A_{i_j}(I-F)^{-1}B)^{t} \Big) & = 0. & \Comment{ $ M_{(s)}^{-1} $ is full rank}  \label{Eq:depend}
\end{align}
Equation \eqref{Eq:depend} 
implies a linear dependence among $s-1$ rows of the matrix $M$.
This contradicts the fact that for each $i=1,2,\ldots,s$,
$M_{(i)}$ is full
rank.
Thus $\left(A_s(I-F)^{-1}B^{t} M_{(s)}^{-1}\right)_i \neq 0$ for $i=1,2,\ldots,s-1$ and the claim follows.

\end{proof}
\begin{remark}
We provide the following  communication-theoretic interpretation of our method of proof above. 
We  may view the computation problem as a MIMO (multiple input multiple output) channel where
the multiple input is given by the vector of symbols generated by the sources, the output is the vector decoded by the receiver, and the channel is given by the network topology and the network code.
Our objective is to choose a channel 
to guarantee the desired output, by way of code design subject to the constraints imposed by network topology.
The channel gain from source $\source_i$ to the receiver is given by the vector  $M_i$ of length $s-1$.
The first part of the proof utilizes the sparse zeros lemma to establish that there exists a choice of 
channels such that the channel between every set of  $s-1$ sources and 
the receiver is invertible. This is similar to the proof of the multicast theorem in \cite{Koetter-Medard-IT-Oct03}.
In the second part of the proof, we recognize that the interference from different sources must also be ``aligned'' at the output for the
receiver to be able to compute the desired function. Accordingly, we have modified the code construction to provide such alignment.
\end{remark}
We now show the existence of a linear function that cannot be computed on a network satisfying the min-cut condition. This network will then be used as a building block to show an analogous result for a larger class of functions.
Let $\transferM_1$ denote the matrix
\begin{align} \label{Eq:T1}
\begin{pmatrix}
1 & 0 & 1 \\
0 & 1 & 0 
\end{pmatrix}
\end{align}
and let $f_1$ denote the corresponding linear function. It is possible to show with some algebra that $T_1 \not \sim (I \; u),$ for any column vector $u$ of units, so that the conclusion of Theorem~\ref{Th:solvable} does not hold. 
Indeed, for the function $f_1$ the opposite conclusion is true, namely $f_1$ cannot be computed over $\Network_1$ using linear codes. This is shown by  the following Lemma. 
%
\begin{lemma} \label{Lemma:gBasis1}
Let $\Network_{\GBASISFIG}$ be the network 
shown in Figure~\ref{Fig:gBasis1} with alphabet $\field{q}$.
We have
\begin{enumerate}
\item $\mbox{\cut{\Network_{\GBASISFIG},T_1}} = 1$.
\item There does not exist a  linear solution for computing $f_1$ in $\Network_1$.
\end{enumerate}
\end{lemma}
\begin{figure}[hht]
\begin{center}
\psfrag{X}{$\source_2$}
\psfrag{Y}{$\source_1$}
\psfrag{Z}{$\source_3$}
\psfrag{T}{$\receiver$}
\psfrag{e1}{$e_1$}
\psfrag{e2}{$e_2$}
\psfrag{e3}{$e_3$}
\psfrag{e4}{$e_4$}
\scalebox{.8}{\includegraphics{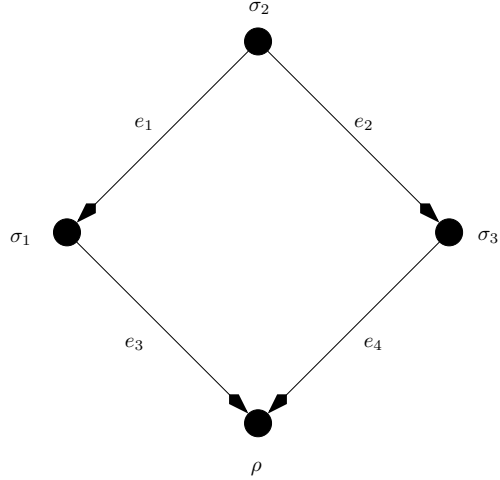}}
\end{center}
\caption{Network  on which there is no linear solution for computing $f_1$.} 
\label{Fig:gBasis1}
\end{figure}
\begin{proof}
That $\mbox{\cut{\Network_{\GBASISFIG},T_1}} = 1$ is easily verified by
considering the cut $C = \{e_3,e_4\}$ which attains the minimum.
We now proceed to show, using Theorem~\ref{Th:gbasisCondition}, that a linear solution does not exist.


We may assume, without loss of generality, that the node $\source_2$ sends its message directly to nodes
$\source_1$ and $\source_3$ (i.e., $x_{1,e_1}=x_{1,e_2}=1$).
The matrices $Z_1,Z_2$, and $Z_3$ over $R$ can then be written as 
\begin{align*}
(T_1)^{t}-M_1 &  = \begin{pmatrix}
(1- x_{1,e_3} x_{e_3,1}) & (0-x_{1,e_3} x_{e_3,2})
\end{pmatrix} \\
(T_2)^{t}-M_2 & = \begin{pmatrix}
0-x_{e_1,e_3} x_{e_3,1}-x_{e_2,e_4} x_{e_4,1} \\ 1-x_{e_1,e_3} x_{e_3,2}- x_{e_2,e_4} x_{e_4,2}
\end{pmatrix}^{t} \\
(T_3)^{t}-M_3 & = \begin{pmatrix}
(1-x_{1,e_4} x_{e_4,1}) & (0-x_{1,e_4} x_{e_4,2})
\end{pmatrix}.
\end{align*}
Consequently, the ideal $J$ is given by
\begin{align*}
J & = \big\langle (1- x_{1,e_3} x_{e_3,1}), \ (0-x_{1,e_3} x_{e_3,2}), \\
  & \quad  \ (0-x_{e_1,e_3} x_{e_3,1}-x_{e_2,e_4} x_{e_4,1}), \\
  &  \quad \ (1-x_{e_1,e_3} x_{e_3,2}- x_{e_2,e_4} x_{e_4,2}), \\
  & \quad \ (1-x_{1,e_4} x_{e_4,1}), \ (0-x_{1,e_4} x_{e_4,2}) 
   \big\rangle.
\end{align*}
%
%
We have 
\begin{align*}
1 & = (1-x_{e_1,e_3}x_{e_3,2} - x_{e_2,e_4}x_{e_4,2}) \\
  & \quad + x_{e_1,e_3} x_{e_3,2} (1- x_{1,e_3} x_{e_3,1}) \\
  & \quad - x_{e_1,e_3}x_{e_3,1} (0- x_{1,e_3} x_{e_3,2} ) \\
  & \quad + x_{e_2,e_4}x_{e_4,2} (1-x_{1,e_4}x_{e_4,1}) \\
  & \quad - x_{e_2,e_4}x_{e_4,1} (0- x_{1,e_4}x_{e_4,2}) \ \in \ J.
\end{align*}
Thus, it follows that $\gbasis{J} = \{1\}$.
By Theorem~\ref{Th:gbasisCondition},
a linear solution does not exist for computing $f_1$ in $\Network_{\GBASISFIG}$.
\end{proof}
We now identify a much larger class of linear functions for  which there exist networks satisfying the min-cut condition but for which linear solutions do not exist.
Let $P$ be an $l \times s-l$ matrix with at least one zero element and
$T \sim (I \ P)$. For each $T$ in this equivalence class we show that there exist a network $\Network$ that does not have a solution
for computing the linear target function corresponding to $T$ but satisfies the cut condition in Lemma~\ref{Lemma:ub}.
The main idea of the proof is to establish that a solution for computing such a function in network $\Network$ implies 
a solution for computing the function corresponding to $T_1$ in $\Network_{\GBASISFIG}$, and then to use 
Lemma~\ref{Lemma:gBasis1}.
\begin{theorem} \label{Th:insolvable}
Consider a linear target function $f$
corresponding to a matrix $\transferM \in \field{q}^{l \times s}$.
If $\transferM \sim (I \ P)$ such that at least one element of $P$ is zero,
then there exists a network $\Network$ such that
\begin{enumerate}
\item $\mbox{\cut{\Network,T}} = 1$. 
\item There does not exist a linear solution  for computing $f$ in $\Network$.
\end{enumerate}
\end{theorem}
\begin{proof}
%
Let $\hat{T} = (I \ P)$ and let $\hat{f}$ denote the corresponding linear target function.
It is enough to show that there exists a network $\Network_P$ 
such that $\mbox{\minCut{\Network_P,f}} = 1$ 
but $\Network_P$ does not have a 
linear solution  for computing $\hat{f}$.
This is because a network $\Network$ that does not
have a solution for computing $T$ is 
easily obtained by renaming the sources in $\Network_P$ as follows:
Since $T \sim (I \ P)$, there exist $Q$ and $\Pi$ such that $T = Q (I \ P) \Pi$.
Let $\kappa$ denote the permutation function on the set $\{1,2,\ldots,s\}$ defined by the permutation matrix
$\Pi^{-1}$.
Obtain the network $\Network$ by relabeling source $\source_i$ in $\Network_P$ as $\source_{\kappa(i)}$.
To see that there does not exist a solution for computing $f$ in $\Network$,
assume to the contrary that a solution exists.
By using the same network code in $\Network_P$, the receiver computes
$$
Q (I \ P) \Pi \ (x_{\kappa(1)},x_{\kappa(2)},\ldots,x_{\kappa(s)})^{t} = Q (I \ P) \ (x_1,x_2,\ldots,x_s)^{t}.
$$
Thus the receiver in $\Network_P$ can compute $\hat{T} x^t$, which is a contradiction.

Now we construct the network $\Network_P$ as claimed.
Since $P$ has at least once zero element, 
there exists a $\tau \in \{l+1,l+2,\ldots,s\}$ 
such that $\hat{\transferM}$ has a zero in $\tau$-th column.
Define
$$
K = \left\{ i \in \{1,2,\ldots,l\}:\hat{T}_{i,\tau} = 1\right\}
$$
Denote the elements of $K$ by 
$$\left\{j_1,j_2,\ldots,j_{\card{K}}\right\}.$$
Let $p$ be an element of $\{1,2,\ldots,l\}-K$ (such a $p$
exists from the fact that the $\tau$-th column contains at least one zero)
and define
$$
\bar{K} = \left\{1,2,\ldots,s\right\} - K - \{\tau,p\}
$$
and denote the elements of $\bar{K}$ by 
$$\left\{j_{\card{K}+1},j_{\card{K}+2},\ldots,j_{s-\card{K}-2}\right\}.$$
Since $\hat{T}$ does not contain an all-zero column, $\card{K} > 0$.
Now, let $\Network_P$ denote the network shown in Figure~\ref{Fig:insolvability} where, 
$\node$ denotes a relay node.
\begin{figure*}[ht]
\begin{center}
\psfrag{sp}{{\large $\source_{p}$}}
\psfrag{st}{{\large $\source_{\tau}$}}
\psfrag{s2}{{\large $\source_{j_1}$}}
\psfrag{s3}{{\large $\source_{j_2}$}}
\psfrag{s4}{{\large $\source_{j_{\card{K}-2}}$}}
\psfrag{s5}{{\large $\source_{j_{\card{K}-1}}$}}
\psfrag{s6}{{\large $\source_{j_{\card{K}}}$}}
\psfrag{s7}{{\large $\source_{j_{\card{K}+1}}$}}
\psfrag{s8}{{\large $\source_{j_{\card{K}+2}}$}}
\psfrag{s9}{{\large $\source_{j_{s-\card{K}-3}}$}}
\psfrag{s10}{{\large $\source_{j_{s-\card{K}-2}}$}}
\psfrag{n}{{\large $\node$}}
\psfrag{r}{{\large $\receiver$}}
\scalebox{.7}{\includegraphics{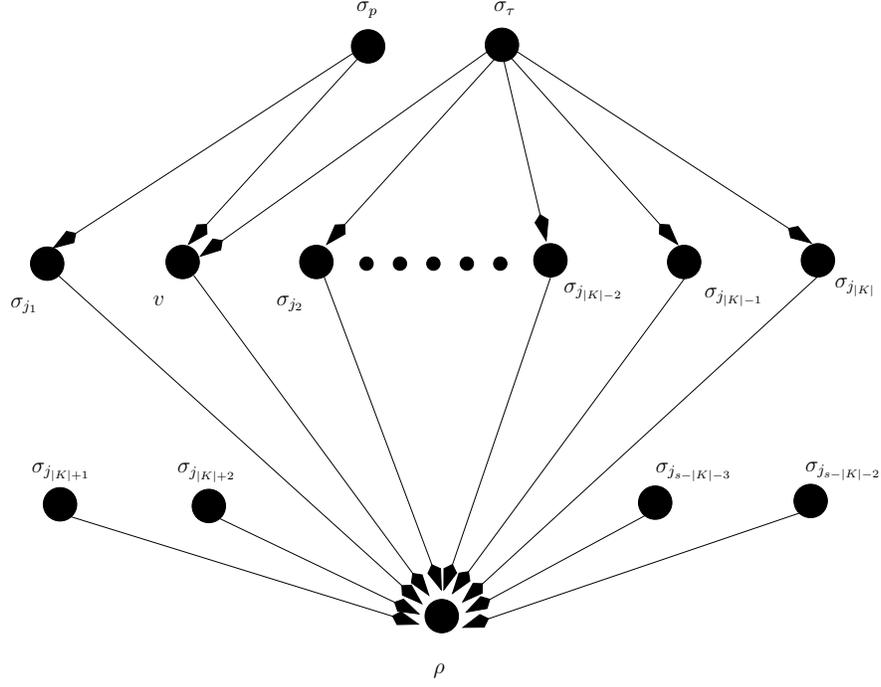}}
\end{center}
\caption{Network $\Network_{P}$ with min-cut  $1$ that does not
have an $\field{q}$-linear solution  for computing $(I \ P)$.} 
\label{Fig:insolvability}
\end{figure*}
It follows from the construction that 
\begin{align} \label{Eq:reducedT}
\begin{pmatrix}
\hat{T}_{j_1,j_1} & \hat{T}_{j_1,p} & \hat{T}_{j_1,\tau} \\
\hat{T}_{p,j_1} & \hat{T}_{p,p} & \hat{T}_{p,\tau} \\
\end{pmatrix} = 
\begin{pmatrix} 
1 & 0 & 1 \\
0 & 1 & 0
\end{pmatrix}
\end{align}
which is equal to the transfer matrix $\transferM_1$ defined in
\eqref{Eq:T1}.

Notice that in the special case when $K=\{j_1\}$ and $\card{\bar{K}}=0$,
the network shown in Figure~\ref{Fig:insolvability} 
reduces to the network shown in Figure~\ref{Fig:reducedSum2}
which is equivalent
to the network $\Network_{\GBASISFIG}$ in Figure~\ref{Fig:gBasis1} with target
function $f_1$. 
Since $\Network_{\GBASISFIG}$ does not have a solution 
for computing $f_1$ by Lemma~\ref{Lemma:gBasis1}, 
we conclude that $\Network_{\GBASISFIG}$ cannot have a solution either.

Similarly, we now show that in the general case, 
if the network $\Network_{P}$ 
has a solution for computing $\hat{f}$,
then such a solution induces a solution for computing $f_1$ in network $\Network_{\GBASISFIG}$, contradicting
Lemma~\ref{Lemma:gBasis1}.
Let there exist  an $n > 0$ for which there is a  linear solution for computing $\hat{f}$ over
 $\Network_P$ using an alphabet over $\field{q^n}$.
In any such solution,
for each $j \in K-\{j_1\}$, 
the encoding function on the edge 
$(\source_j,\receiver)$ 
must be of the form 
\begin{align} \label{Eq:edgeForm}
\beta_{1,j} \sourceVec{\source_{j}} + \beta_{2,j} \sourceVec{\source_{\tau}}
\end{align}
for some $\beta_{1,j},\beta_{2,j} \in \field{q^n}$.
Since $(\source_j,\receiver)$ is the only path from 
source $\source_j$ to the receiver, it is obvious that 
$\beta_{1,j} \neq 0$.

We define the map $\alpha$ as follows.
Let $\sourceVec{\source_{j_1}},\sourceVec{\source_{p}},\sourceVec{\source_{\tau}}$ be arbitrary elements of $\field{q^n}$ and 
let
\begin{align} \label{Eq:souceChoice}
\sourceVec{\source_j} = \begin{cases}
0 \; & \mbox{for} \; j \in \bar{K} \\
-(\beta_{1,j})^{-1}\beta_{2,j} \sourceVec{\source_{\tau}}
 \;  & \mbox{for} \; j \in K-\{j_1\}.
\end{cases} 
\end{align}
Note that $\alpha$ has been chosen such that 
for any choice of $\sourceVec{\source_{j_1}},\sourceVec{\source_{p}}$, and $\sourceVec{\source_{\tau}}$, every edge 
$e \in \inEdges{\receiver} - \{(\source_{i_1},\receiver),(\node,\receiver)\}$ carries the zero vector.
Furthermore, for the above choice of $\alpha$,
the target function associated with $\hat{T}$ reduces to
\begin{align} \label{Eq:receiverComputes}
\left( \sourceVec{\source_{1}}+ \hat{T}_{1,\tau}\sourceVec{\source_{\tau}},\sourceVec{\source_{2}}+ \hat{T}_{2,\tau}\sourceVec{\source_{\tau}},\ldots,\sourceVec{\source_{l}}+ \hat{T}_{l,\tau}\sourceVec{\source_{\tau}}\right).
\end{align}
Substituting $\hat{T}_{j_1,\tau} =1$ and 
$\hat{T}_{p,\tau} = 0$ in \eqref{Eq:receiverComputes}, 
it follows that the receiver can compute
$$
\left( \sourceVec{\source_{j_1}} + \sourceVec{\source_{\tau}},\sourceVec{\source_{p}}\right)
$$
from the vectors received on edges $(\source_{i_1},\receiver)$ and $(\node,\receiver)$.
Consequently, it follows that there exist a  linear
solution over $\field{q^n}$ for computing the linear
target function associated with the transfer matrix
$$
\begin{pmatrix}
\hat{T}_{j_1,j_1} & \hat{T}_{j_1,p} & \hat{T}_{j_1,\tau} \\
\hat{T}_{p,j_1} & \hat{T}_{p,p} & \hat{T}_{p,\tau} \\
\end{pmatrix}
$$
\begin{figure}[ht]
\begin{center}
\psfrag{sil}{{\large $\source_{p}$}}
\psfrag{spl}{{\large$\source_{\tau}$}}
\psfrag{si2}{{\large$\source_{j_1}$}}
\psfrag{n}{{\large$\node$}}
\psfrag{r}{{\large$\receiver$}}
\scalebox{.7}{\includegraphics{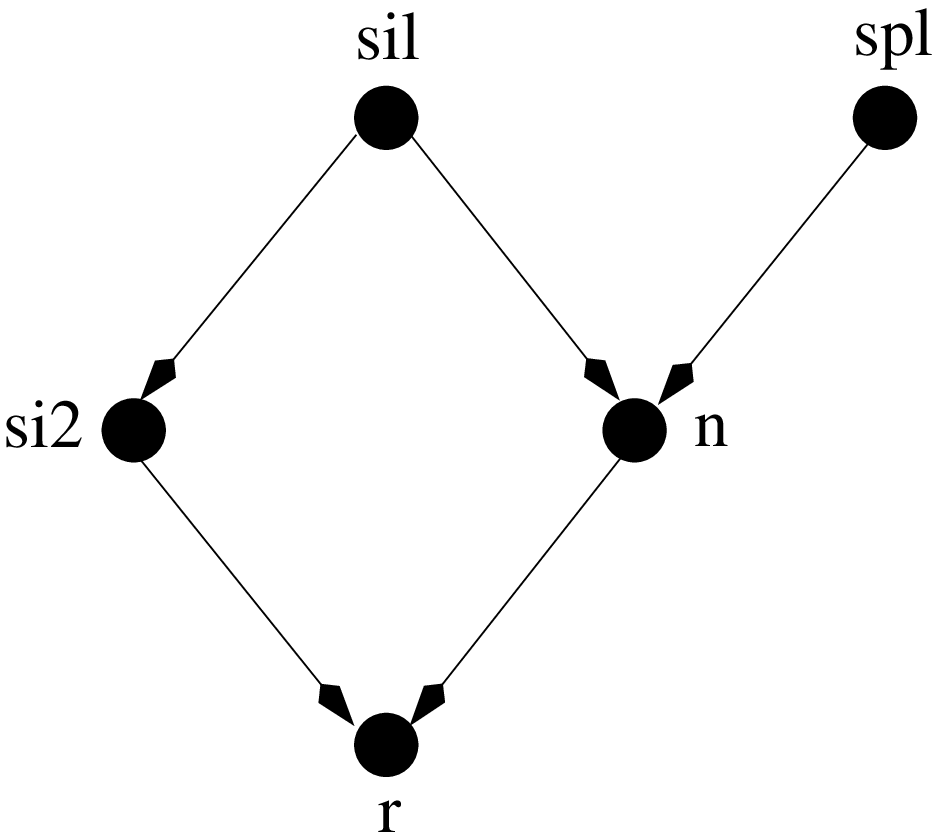}}
\end{center}
\caption{Subnetwork of $\Network_{P}$ used to show the equivalence between solving network $\Network_{P}$ and solving network $\Network_{\GBASISFIG}$.
} 
\label{Fig:reducedSum2}
\end{figure}
in the network shown in Figure~\ref{Fig:reducedSum2}.
It is easy to see that 
the existence of such a code implies a scalar linear
solution  for computing $f_1$ in $\Network_{\GBASISFIG}$.
This establishes the desired contradiction.

Finally, we show that $\mbox{\cut{\Network,T}} = 1$.
Let $C \in \cuts{\Network}$ be a cut such that $K_C \subset K \cup \{p,\tau\}$ 
(i.e, $C$ separates sources  from only the top and middle rows in the network $\Network_P$).
We have the following two cases:
\begin{enumerate}
\item
If $\source_{\tau} \notin K_C$, then it is easy to see that $\card{C} \ge \card{K_C}$.
Similarly, if $\source_{\tau} \in K_C$ and $\source_{p} \notin K_C$, then again  $\card{C} \ge \card{K_C}$.
Consequently, we have
\begin{align}
\frac{\card{C}}{\rank{T_{K_C}}} & \ge \frac{\card{C}}{\card{K_C}} & \Comment{$\rank{T_{K_C}} \le \card{K_C}$} \nonumber \\
																& \ge 1. & \Comment{$\card{C} \ge \card{K_C}$} \label{Eq:minCutForNP1}
\end{align}
%
\item  If $\source_{\tau} \in K_C$ and  $\source_p \in K_C$, then from Figure~\ref{Fig:reducedSum2}, $\card{C} = \card{K}+1$ and $K_C = K \cup \{p,\tau\}$.
Moreover, the index set $K$ was constructed such that 
\begin{align} \label{Eq:linearRelation}
\hat{T}_{\tau} = \sum_{i \in K} \hat{T}_{i,\tau} \hat{T}_{i}.
\end{align}
Consequently, we have
\begin{align}
\rank{T_{K_C}} & = \rank{T_{K \cup \{p,\tau\}}}  & \Comment{$K_C = K \cup \{p,\tau\}$} \nonumber \\
							 & \le \card{K} + 1							& \Comment{\eqref{Eq:linearRelation}} \nonumber \\
							 & = \card{C}. \label{Eq:minCutForNP2}
\end{align}
\end{enumerate}
From \eqref{Eq:minCutForNP1} and \eqref{Eq:minCutForNP2}, we conclude that if $K_C \subset K \cup \{p,\tau\}$, then
\begin{align}
\frac{\card{C}}{\rank{T_{K_C}}} & \ge 1. \label{Eq:cutValueCase1}
\end{align}
For an arbitrary cut $C \in \cuts{\Network}$, let $c_{\bar{K}}$ denote the number of sources 
in $\bar{K}$ that are separated from the receiver by $C$ (i.e,   $c_{\bar{K}} = \card{K_C \cap \bar{K}}$).
We have
\begin{align}
\frac{\card{C}}{\rank{T_{K_C}}} & = \frac{\card{C}-c_{\bar{K}}+c_{\bar{K}}}{\rank{T_{K_C}}} \nonumber \\
																& \ge \frac{\card{C}-c_{\bar{K}}+c_{\bar{K}}}{\rank{T_{K_C-\bar{K}}}+c_{\bar{K}}} 
\label{Eq:minCutSplit}
\end{align}
Since each source in $\bar{K}$ is directly connected to the receiver,
$\card{C}-c_{\bar{K}}$ is equal to the number of edges in $C$ separating the sources in $K_C-\bar{K}$
from the receiver.
Consequently, from \eqref{Eq:cutValueCase1}, it follows that 
\begin{align} \label{Eq:particalLwrBnd}
\frac{\card{C}-c_{\bar{K}}}{\rank{T_{K_C-\bar{K}}}} & \ge 1.
\end{align}
Substituting \eqref{Eq:particalLwrBnd} in \eqref{Eq:minCutSplit}, we conclude that for all $C \in \cuts{\Network}$
$$
\mbox{\cut{\Network,T}} \ge 1.
$$
Since the  edge $(\source_{j_{\card{K}+1}},\receiver)$
disconnects the source $\source_{j_{\card{K}+1}}$ from the receiver, 
$\mbox{\cut{\Network,T}} \le 1$ is immediate and the proof of the theorem is now complete.
\end{proof}
We now consider the case in which the source alphabet is over the binary field. In this case, we have that the two function classes identified by Theorems~\ref{Th:solvable} and \ref{Th:insolvable} are complements of each other, namely either $T \sim (I \ \mb{1})$ or $T \sim (I \ P)$ with $P$ containing at least one zero element. 

%
\begin{theorem} \label{Lemma:binaryField}
Let $l \notin \{1,s\}$ and let $\transferM \in \field{2}^{l \times s}$.
If   $T \nsim (I \ \mb{1})$,
then there exists an $l \times (s-l)$ matrix $P$ such that $P$ 
has at least one zero element and $T \sim (I \ P)$.
\end{theorem}
\begin{proof}
Since $T$ is assumed to have a full row rank, $T \sim (I \ \bar{P})$ for 
some $l \times (s-l)$ matrix $(I \ \bar{P})$ over $\field{2}$.
If $\bar{P}$ has $0$'s, then we are done.
Assume to the contrary  that $\bar{P}$ is a matrix of non-zero elements.
We only need to consider the case when $(s-l) > 1$ (since $T \nsim (I \ \mb{1})$).
For $i=1,2,\ldots,l-1$, 
let $\phi^{(i)}$ denote the $i$-th column vector of the $l \times l$ identity matrix.
Define 
$Q = (\phi^{(1)} \phi^{(2)} \cdots \phi^{(l-1)} \  \mb{1})$
and let $\Pi$ be a permutation matrix that interchanges the
$l$-th and $(l+1)$-th columns and leaves the remaining columns unchanged.
It is now easy to verify that 
\begin{align}
Q \ ( I \ \bar{P}) \ \Pi & = (Q \ Q \bar{P}) \ \Pi \nonumber \\
											   & = (I \ P) \label{Eq:equivalenceConst}
\end{align}
where $P$ is an $l \times s-l$ matrix with at least one $0$ element:
for $i \in \{1,2,\cdots,l-1\}$ 
\begin{align*}
P_{i,2} & = (Q\bar{P})_{i,2} \\
			  & = (Q \mb{1})_i \\
			  & = 1 + 1 \\
			  & = 0. 
\end{align*}
%
Thus, $( I \ \bar{P}) \sim (I \ P)$ and by transitivity we conclude that
$T \sim (I \ P)$ which proves the claim.
\end{proof}
%
%
\section{Conclusion} \label{Sec:conclusions}
%

We wish to mention the following open problems 
arising from this work.
\begin{itemize}
\item Is there a graph-theoretic condition that allows to determine whether  a given network is solvable with reference to a given linear function? We have provided an algebraic test in terms of the \grobner basis of a corresponding ideal, but we wish to know whether there is there an algorithmically more efficient test.  
\item We showed that $\mbox{\cut{\Network,T}} = 1$ is not sufficient to guarantee solvability for a certain class of linear
functions. A possible direction of future research is to ask whether there is a constant $c$ such that $\mbox{\cut{\Network,T}} \ge c$  guarantees solvability.
Alternatively, for every constant $c$, does there exist a network $\Network$ and a matrix $T$ such 
that $\mbox{\cut{\Network,T}} \ge c$ and $\Network$ does not have a linear solution for computing the linear
target function associated with $T$?
\end{itemize}

\clearpage
\begin{appendix}
\section{Appendix}
\begin{lemma} \label{Lemma:equivalenceLinearFunctions}
Let $\transferM \in \field{q}^{l \times s}$.
If $u \in \field{q}^{s-1}$ is a column vector of non-zero elements and $\transferM \sim (I \ u)$, 
then there exists a full rank matrix $Q$ and a column vector $u'$ of non-zero elements over $\field{q}$ such 
that $T = Q \ (I \ u')$.
\end{lemma}
\begin{proof}
Let $Q$ denote the matrix obtained by collecting the first $(s-1)$ columns of $T$.
We will first show that the matrix $Q$ is full-rank.
After factoring out $Q$, we then prove that the last column must
have non-zero entries.

Since $T \sim (I \ u)$, there exists a full-rank matrix $\bar{Q}$ and a permutation matrix 
$\bar{\Pi}$ such that
\begin{align}
T & = \bar{Q} \ (I \ u) \ \bar{\Pi} \nonumber \\
  & = (\bar{Q} \ \bar{Q} u) \ \bar{\Pi}. \label{Eq:factorT}
\end{align}
From \eqref{Eq:factorT}, the columns of $Q$ are constituted by the columns of $\bar{Q}$ in which case $Q$ is full-rank, or columns of $Q$ contains $(s-2)$ columns of $\bar{Q}$ and $\bar{Q} u$.
We will now show that 
the vector $\bar{Q} u$ cannot be written as a linear combination of any set of $s-2$ column vectors
of $\bar{Q}$.
Assume to the contrary that there exist $a_j \in \field{q}$ for $j \in \{1,2,s-2\}$ such that
\begin{align}
\bar{Q} u & = \sum_{j=1}^{s-2} a_j \bar{Q}_{j} \label{Eq:dependence}
\end{align}
where $\bar{Q}_{j}$ denotes the $j$-th column of $\bar{Q}$.
Let $a$ denote the vector such that $a_{j} = a_j, j=1,2,\ldots s-2$, and $a_{s-1} = 0$.
We have
\begin{align}
u-a           & \neq 0 & \Comment{$u_{s-1}\neq 0$ and $a_{s-1} = 0$} \nonumber \\
\bar{Q} (u-a) & = 0 & \Comment{\eqref{Eq:dependence}}. \label{Eq:contradiction}
\end{align}
\eqref{Eq:contradiction} contradicts the fact that $\bar{Q}$ is full-rank.
Hence  $a_i$'s satisfying \eqref{Eq:dependence} do not exist and consequently, 
$Q$ is a full-rank matrix.
We now have
$$
T = Q(I \ u')
$$
where $u' = Q^{-1} T_{s}$ and hence $T \sim (I \ u')$. 
Furthermore, $T \sim ( I \ u)$ and $ T \sim (I \ u')$ implies that $( I \ u)   \sim (I \ u')$.
Thus, there exists a full-rank matrix $P$ and a permutation matrix $\Pi$ 
such that
\begin{align}
%
( I \ u) & = P  \ (I \ u') \ \Pi & \nonumber \\ 
							& = (P \ P u') \ \Pi. \label{Eq:eqiv}
\end{align}
Let $\phi^{(i)}$ denote the $i$-th column of $I$.
It follows from \eqref{Eq:eqiv} that either 
$(a)$ $P u' = u$ and $P$ itself is an $(s-1) \times (s-1)$ permutation matrix, or
$(b)$ For some $j \in \{1,2,\ldots,s-1\}$, 
$j$-th column of $P$ is $u$,  
and the remaining columns must constitute the $s-2$ columns
$\phi^{(1)},\phi^{(2)},\ldots,\phi^{(\tau-1)},\phi^{(\tau+1)},\phi^{(s-1)}$ of $I$ for some $\tau$.
If $(a)$ is true, then $u'=P^{-1}  u$ and the elements of $u'$ are non-zero 
since $P^{-1}$ is another permutation matrix.
If $(b)$ is true, then $P u' = \phi^{(\tau)}$ and it must be that $u'_j \neq 0 $
(if $u'_j = 0 $, then $(P u')_{\tau} = 0$ which contradicts $P u' = \phi^{(\tau)}$).
Let $L = \{i : i \neq j, \ \mbox{and} \ u'_i \neq 0\}$.
We must have
\begin{align} \label{Eq:1}
\phi^{(\tau)} = u'_j  u + \sum_{i \in D} u'_i \  \phi^{(j_i)}.
\end{align}
If we denote the number of non-zero entries in a vector $u$ by
$\card{u}$,
then we have
\begin{align}
 1 & = \card{\phi^{(\tau)}} \nonumber \\
   & \ge \card{u'_j  u} - \card{D} & \Comment{\eqref{Eq:1}} \nonumber \\
   & = (s-1) - \card{D} \nonumber \\
   & \ge 1 & \Comment{$\card{D} \le s-2$} \label{Eq:sandwich}
\end{align} 
From \eqref{Eq:sandwich}, it follows that $\card{D} = s-2$ and consequently that
every element of $u'$ is non-zero. The proof of the lemma is now complete.
\end{proof}

\end{appendix}
%
%

%

\begin{thebibliography}{10}
%
%
\bibitem{Ahlswede-Cai-Li-Yeung-IT-Jul00}
R. Ahlswede, N. Cai, S.-Y. R. Li, and R. W. Yeung,
``Network information flow'',
\textsl{IEEE Transactions on Information Theory},
vol. IT-46, no. 4, pp. 1204--1216, July 2000.
%
\bibitem{Koetter-Medard-IT-Oct03}
R. Koetter and M. M\'{e}dard,
``An algebraic approach to network coding,''
\textsl{IEEE/ACM Transactions on Networking},
vol. 11, no. 5, pp. 782--795, Oct. 2003.
%
\bibitem{Kumar1} 
H. Kowshik and P. R. Kumar, 
``Zero error function computation in sensor networks'', 
\textsl{Proceedings of the IEEE Conference on Decision and Control}, 
2009.
%
\bibitem{RaiDey2009} 
B. K. Rai, and B. K. Dey,
``Sum-networks: System of polynomial equations, unachievability of coding capacity, reversibility, insufficiency of linear network coding,''
http://arxiv.org/abs/0906.0695, 2009.
%
\bibitem{ramam} 
A. Ramamoorthy,
``Communicating the sum of sources over a network,''
\textsl{Proceedings of the IEEE International Symposium on Information Theory}, 
Toronto, Canada, 2008.
%
\bibitem{Ma}
N. Ma, P. Ishwar, and P. Gupta, 
``Information-theoretic bounds for multiround function computation in collocated networks,'' 
\textsl{Proceedings of the IEEE International Symposium on Information Theory},
pp. 2306--2310, 2009.
%
\bibitem{Nazer}
B. Nazer and M. Gastpar,
``Computing over multiple-access channels,''
\textsl{IEEE Transactions on Information Theory},
vol. 53, pp. 3498--3516, Oct. 2007.
%
\bibitem{ken1}
R. Dougherty, C. Freiling, and K. Zeger,
``Insufficiency of linear coding in network information flow,''
\textsl{IEEE Transactions on Information Theory},
vol. 51, no. 8, pp. 2745-2759, August 2005.
%
\bibitem{ken2}
R. Dougherty, C. Freiling, and K. Zeger,
``Linear network codes and systems of polynomial equations'',
\textit{IEEE Transactions on Information Theory}
vol. 54, no. 5, pp. 2303-2316, May 2008.
%
\bibitem{computing1}
R. Appuswamy, M. Franceschetti, N. Karamchandani, and K. Zeger,
``Network coding for computing: Cut-set bounds'',
\textit{to appear in  IEEE Transactions on Information Theory}, Feb. 2011.
%
\bibitem{computing2}
R. Appuswamy, M. Franceschetti, N. Karamchandani, and K. Zeger,
``Network coding for computing: Linear codes'',
\textit{submitted to  IEEE Transactions on Information Theory},  2010.
%
%
\bibitem{Schwartz}
J. T. Schwartz, ``Fast probabilistic algorithms for verification of polynomial identities'',
\textsl{J. ACM.},
vol. 27, pp. 701-717, 1980.
%
\bibitem{Hungerford}
T.W. Hungerford, ``Algebra'', \textsl{Springer-Verlag}, 1997.
%
\end{thebibliography}
\end{document}